\documentclass[letterpaper, 10 pt, conference]{ieeeconf}  

\IEEEoverridecommandlockouts                              

\usepackage{graphics} 
\usepackage{epsfig} 
\usepackage{mathptmx} 
\usepackage{times} 
\usepackage[utf8]{inputenc}
\usepackage{amsmath, amssymb}
\usepackage{siunitx}
\usepackage{graphicx, import}
\usepackage[dvipsnames]{xcolor}
\usepackage{subcaption}
\usepackage{tikz}
\usepackage{pgfplots}
\usepackage{algorithm}
\usepackage{enumerate}
\usepackage{algorithmic}
\usepackage{booktabs}
\usepackage{makecell}
\usepackage{soul}
\usetikzlibrary{pgfplots.groupplots, positioning, decorations.pathreplacing, arrows.meta, decorations.markings}
\newtheorem{theorem}{Theorem}

\usepackage[style=ieee,backend=bibtex,bibencoding=utf8,natbib=true,url=false,doi=false,isbn=false,date=year,maxbibnames=10,minbibnames=10,maxcitenames=2,mincitenames=1,giveninits=true]{biblatex}
\title{\LARGE \bf
Towards Event-Triggered NMPC for Efficient 6G Communications:\\Experimental Results and Open Problems
}

\author{Jens Püttschneider, Julian Golembiewski, Niklas A. Wagner, Christian Wietfeld, Timm Faulwasser
\thanks{This work has been partly funded by the Federal Ministry of Education and Research (BMBF) via the project \emph{6GEM} under funding reference 16KISK038, the Deutsche Forschungsgemeinschaft (DFG, German Research Foundation) under project number 508759126, and by the Ministry for Economic Affairs, Industry, Climate Action and Energy (MWIKE) of the State of North Rhine-Westphalia via the project \emph{5hine} under funding reference 005-2108-0074.}
\thanks{\textbf{Jens Püttschneider}, \textbf{Julian Golembiewski}, Institute of Energy Systems, Energy Efficiency and Energy Economics, TU Dortmund University, Dortmund, Germany
        {\tt\small \{jens.puettschneider, julian.golembiewski\}@tu-dortmund.de}}%
\thanks{\textbf{Niklas A. Wagner}, \textbf{Christian Wietfeld}, Communication Networks Institute (CNI), TU Dortmund University, Dortmund, Germany
        {\tt\small \{niklas.wagner, christian.wietfeld\}@tu-dortmund.de}}%
\thanks{\textbf{Timm Faulwasser}, Institute of Control Systems, Hamburg University of Technology, Hamburg, Germany
        {\tt\small timm.faulwasser@ieee.org}}%
}

\begin{document}

\maketitle
\thispagestyle{empty}
\pagestyle{empty}

\begin{abstract}
    Networked control systems enable real-time control and coordination of distributed systems, leveraging the low latency, high reliability, and massive connectivity offered by 5G and future 6G networks. 
    Applications include autonomous vehicles, robotics, industrial automation, and smart grids.
	Despite networked control algorithms admitting nominal stability guarantees even in the presence of delays and packet dropouts, their practical performance still heavily depends on the specific characteristics and conditions of the underlying network.
	To achieve the desired performance while efficiently using communication resources, co-design of control \textit{and} communication is pivotal.
	Although periodic schemes, where communication instances are fixed, can provide reliable control performance, unnecessary transmissions, when updates are not needed, result in inefficient usage of network resources.
	In this paper, we investigate the potential for co-design of model predictive control and network communication.
	To this end, we design and implement an event-triggered nonlinear model predictive controller for stabilizing a Furuta pendulum
	communicating over a tailored open radio access network 6G research platform.
    We analyze the control performance as well as network utilization under varying channel conditions and event-triggering criteria. Additionally, we analyze the network-induced delay pattern and its interaction with the event-triggered controller.
    Our results show that the event-triggered control scheme achieves similar performance to periodic control with reduced communication demand.
\end{abstract}
\section{INTRODUCTION}
\label{sec:introduction}

Networked Control Systems (NCS), enabling distributed wirelessly connected plant and controller architectures, are one major use case enabled by 5G and future 6G networks \cite{giordani2020toward}.
However, despite the advances in network development, delays, packet loss, limited bandwidth, energy consumption, and reliability remain a challenge for NCS and need to be addressed by co-design of control and communication~\cite{zhao2018toward}.
Co-design strategies include delay compensation \cite{varutti2014model}, acknowledgments for handling packet losses \cite{dolk2015dynamic}, energy-efficient communication policies \cite{varma2019energy}, and event-based control \cite{heemels2012introduction}.

Event-based control reduces the communication demand while maintaining the control performance \cite{heemels2012introduction} and thus enables networked control with limited network bandwidth. For an overview of event-based control we refer to \cite{heemels2012introduction,peng2018survey}.
Event-based control encompasses two primary approaches: event-triggered control and self-triggered control. 
In self-triggered control, the controller calculates the control input and the next recalculation time~\cite{marti2002improving,GOMMANS201559}.
In event-triggered control, on the other hand, an input is applied in sample-and-hold fashion until a predefined triggering condition is met.
This triggering criterion can either be computed by the sensor based on the state of the plant~\cite{heemels2012introduction,varutti2009event,wang2015,liu2019codesign}, or the event triggering mechanism can be placed at the controller with a criterion evaluated based on the control input~\cite{7922495}.
This paper uses an event-triggering mechanism placed at the sensor, which has the advantage of saving communication demand in both directions, sensor to controller and controller to actuator.

Event-based control is especially powerful in combination with Model Predictive Control (MPC), which computes its feedback by optimizing an open-loop trajectory that can be applied until a triggering criterion, based on the deviation of the state to the predicted trajectory is met~\cite{varutti2009event}. Despite an extensive theoretical analysis of event-triggered MPC and event-triggered control more broadly, most papers illustrate the proposed schemes via numerical simulations only~\cite{varutti2009event,7922495,wang2015,Balaghiinaloo2020,yang2018event}.

As highlighted in~\cite{Balaghiinaloo2020}, experimental studies are crucial for advancing research in NCS.
They provide valuable insights into control performance and communication resource requirements for control applications in both current and future wireless networks.
Only a limited number of works explore real hardware implementations: \cite{zhou2023, wang2021etmpc} implement their applications using wired experimental setups. 
Likewise, \cite{lehmann2011extension} evaluates event-based control in a real-world scenario with idealized communication assumptions. A hardware-in-the-loop distributed model predictive control experiment is presented in~\cite{grafe2022event}, using a simulated wireless network. 	
In~{\cite{altaf2011}}, an event-triggered scheme for a nonlinear system is implemented and its performance is compared to that of a periodic controller.
Plant and controller communicate over the contention-based low-power IEEE 802.15.4 standard for wireless personal area networks. 
The wireless personal area network ZigBee is also being used in \cite{9611090} for the event-triggered control networked control of autonomous surface vehicles swarms.
The work in \cite{8025403} implements ET-MPC for vehicle platooning using the IEEE 802.11p Wireless Local Area Network (WLAN) for vehicular communications.
In our previous work~{\cite{overbeck2024data}} we evaluated the impact of different radio resource scheduling strategies on delay and control performance of a delay-compensated periodic Nonlinear Model Predictive Control (NMPC) implementation. 

Extending our previous work, this paper investigates the potential of event-triggered control via 5G and future 6G networks.
To the best of our knowledge, this work is the first to implement event-triggered NMPC within a 5G network. 
The advantages of 5G and future 6G networks over contention-based networks, such as IEEE 802.11 (WLAN), for control applications include network slicing that reduces interference and congestion, better scalability, and reliably low latencies even under high network loads {\cite{9446078, Arendt2024DistributedPerformance}}.
We experimentally validate the controller and analyze the trade-off between control performance and network load considering varying triggering criteria and channel conditions in a state-of-the-art network environment. 
Specifically, we use a 5G Open Radio Access Network (O-RAN) that serves as a research platform towards future 6G communications \cite{Wagner2024a}.

The remainder of this paper is structured as follows: Section \ref{sec:problem_statement} formally defines the problem and outlines the design of the proposed controller, including a stability analysis. Section \ref{sec:experiments} provides details of the experimental setup, implementation, and results, along with a discussion of the findings. Finally, Section \ref{sec:conclusion} concludes the paper and discusses potential directions for future work.
\section{PROBLEM STATEMENT \& CONTROLLER DESIGN}
\label{sec:problem_statement}
We consider the task of controlling a nonlinear plant by an event-triggered controller connected via a communication network.
Consider the continuous-time nonlinear system
\begin{equation*}
	\dot{x}(t) = f(x(t),u(t)), \quad x(0)  = x_0,
\end{equation*}
with state constraints $x(t)\in \mathbb{X} \subseteq \mathbb{R}^{n}$ and admissible inputs $u(t) \in \mathbb{U} \subseteq \mathbb{R}^{m}$ for all $t \in \mathbb{R}$. 
In this work, we assume full state feedback is available.

\subsection{Event-Triggered NMPC Scheme}

The plant is connected to the networked controller as shown in Figure \ref{fig:networked_control_system}. In this setup, the sensor and actuator are connected to the plant and communicate with the controller via an O-RAN 6G research network. Additionally, we assume a connection between the actuator and the sensor, enabling the sensor to receive the input trajectories needed for event-triggering. Section \ref{sec:comm_setup} gives a detailed description of the network. The message sequence diagram is shown in Figure \ref{fig:message_sequence_diagram}.
At sampling time $t_k$, the sensor measures the state of the plant $x(t_k)$.
Based on the latest input trajectory and the corresponding state measurement, the sensor calculates the predicted state $\hat{x}(t_k)$.
If the deviation of predicted state and current measurement exceeds the event-triggering threshold $\varepsilon$, i.e., if $\lVert x(t_k) - \hat{x}(t_k) \rVert \geq \varepsilon$, the sensor sends the current measurement and its timestamp to the controller.

The controller receives the message with the varying uplink delay $\tau_{\mathrm{UL}}(t_k)$. 
Then, calculating the control law takes $\tau_{\mathrm{C}}$, and finally, the control trajectory is communicated back to the actor and applied to the plant, which introduces a delay of $\tau_{\mathrm{DL}}$. 

The assumption of constant controller and downlink delay allows for measuring their values in a preparatory step, which are then used for the delay compensation.
The experimental delay analysis findings in Section \ref{sec:experiments} confirm this assumption.
Thus, the overall round-trip time
\begin{equation*}
	\tau_{\mathrm{RTT}}(t_k) = \tau_{\mathrm{UL}}(t_k) + \tau_{\mathrm{C}} + \tau_{\mathrm{DL}}
\end{equation*}
is known after the packet has arrived at the controller at time $t_k+\tau_{\mathrm{UL}}(t_k)$ since the only time-varying component has realized.
It can be calculated by the controller based on the timestamp of the received state, assuming time synchronization between plant and controller.
We assume the RTT is bounded by $\tau_{\mathrm{RTT}}(t_k) \leq \bar{\tau}_\mathrm{RTT}$. With this the sensor requests an update of the input trajectory if the state deviation does not exceed the threshold and the input trajectory is about to expire in $\bar{\tau}_\mathrm{RTT}$. Using $\bar{\tau}_\mathrm{RTT}$, the sensor can send the new state to the controller such that its responses arrives before the currently available input trajectories expire.

The controller can then compensates the delay by forward predicting the state over the round-trip time (RTT) $\tau_{\mathrm{RTT}}(t_k)$
\begin{equation}
	\hat{x}(t_k + \tau_{\mathrm{RTT}}(t_k)) = x(t_k) + \int_{t_k}^{t_k + \tau_{\mathrm{RTT}}(t_k)} f(x(t), u(t)) \;\mathrm{d}t.
	\label{eq:delay_compensation}
\end{equation}
The input trajectory required for the delay compensation is stored in a buffer at the controller.
The key challenge in designing a network controller with stability guarantees is the input consistency: The input trajectory used for the delay compensation has to match the one applied to the plant in the closed loop in spite of delays and packet losses \cite{varutti2009event,varutti2014model}.

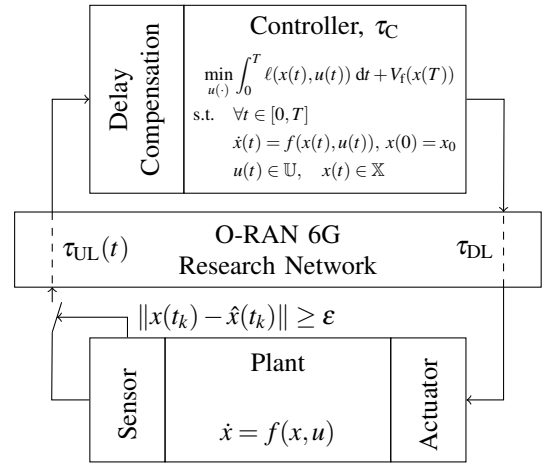
\begin{figure}[ht]
	\centering
	\begin{tikzpicture}[
        font=\sffamily,
        box/.style={draw=none, color=black, fill=Gray, fill opacity=0.3, text opacity=1, rectangle, minimum width=3cm, minimum height=1cm, align=center},
        arrow/.style={-Stealth, thick},
        doublearrow/.style={Stealth-Stealth, thick}
        ]
		\node[box, minimum width=3.75cm, minimum height=3.1cm, align=center] (controller) at (0.625, 2.6) {
			Controller, $\tau_{\mathrm{C}}$ \\[0.2cm]
        \scalebox{0.7}{$
        \begin{aligned}
		\min_{u(\cdot)} &\int_{0}^{T} \ell(x(t),u(t)) \; \mathrm{d}t + V_{\mathrm{f}}(x(T)) \\
		\text{s.t.} \quad & \forall t \in [0, T] \\
		&\dot{x}(t) = f(x(t), u(t)), \; x(0) = x_0 \\
		& u(t) \in \mathbb{U}, \quad x(t) \in \mathbb{X}
        \end{aligned}
        $}
		};
		\node[box, minimum height=1.25cm, minimum width=3.1cm, rotate=90, align=center] (delay) at (-1.885, 2.6) {Delay\\Compensation \eqref{eq:delay_compensation}};
		
		\node[box, minimum width=7cm, minimum height=1.2cm, align=center] (network) at (0, 0) {O-RAN 6G\\Research Network};
		
		\node[box, minimum width=3cm, minimum height=1.8cm, align=center] (system) at (0, -2.6) {
			Plant \\[0.25cm]
			$\dot{x} = f(x, u)$
		};
		\node[box, minimum width=1.8cm, minimum height=1cm, align=center, rotate=90] (actuator) at (2, -2.6) {Actuator};
        
        \node[box, minimum width=1.8cm, minimum height=1cm, align=center, rotate=90] (sensor) at (-2, -2.6) {Sensor};

        \draw[dashed, color=black] (delay.south west) -- (delay.south east);
        \draw[dashed, color=black] (sensor.south west) -- (sensor.south east);
        \draw[dashed, color=black] (actuator.north west) -- (actuator.north east);
    
        \coordinate (switch) at (-3.0, -1.30);
		\draw[-, thick] (sensor.north) |- ++(-0.49, 0) -- (switch);
        \draw[-, thick] (switch) -- ++ (0.125, 0.4);
        \coordinate (switch_top) at (-3.0, -0.90);
        \coordinate (switch_tap) at (-2.9375, -1.1); 
        \draw[arrow] (switch_top) --  (-3.0, -0.6);
        \draw[arrow] (sensor.east) |- (switch_tap); 
        \node at (-2, -1.25) [anchor=west] {$\lVert x(t_k) - \hat{x}(t_k) \rVert \geq \varepsilon$};
		\draw[arrow, dashed, color=black] (-3.0, -0.6) -- (-3.0, 0.6);
		\draw[arrow] (-3.0, 0.6)  -- ++ (0, 1.5) |- (delay.north);
		\node[color=black] at (-3.0, 0) [anchor=west] {$\tau_{\mathrm{UL}}(t)$};
		
		\draw[arrow] (controller.east) |- ++(0.5, 0) -- (3.0, 0.6);
		\draw[arrow, dashed, color=black]  (3.0, 0.6) -- (3.0, -0.6);
		\draw[arrow] (3.0, -0.6)  -- ++ (0, -1.5) |- (actuator.south);
		\node[color=black] at (3.0, 0) [anchor=east] {$\tau_{\mathrm{DL}}$};
	\end{tikzpicture}
	\caption{Networked Control System.}
	\label{fig:networked_control_system}
\end{figure}

\begin{figure}
	\centering
	\resizebox{\columnwidth}{!}{ 
	\begin{tikzpicture}[
			bracket/.style={
				solid,             
				postaction={decorate}, 
				decoration={ 
					markings,
					mark=at position 0 with {\draw[yshift=-3pt] (0,0) -- (0, 3pt);}, 
					mark=at position 1 with {\draw[yshift=-3pt] (0,0) -- (0, 3pt);}  
				}
			}
		]
		
		\draw[dashed] (-2,0) -- (-2,-10.0);
		\draw[dashed] (2,0) -- (2,-10.0);
		
		\node at (-2, 0.5) {Plant};
		\node at (2, 0.5) {Controller};
		
		\draw[thick] (-2.1, -0.5) -- (-1.9, -0.5) node[left, left=15pt] {$t_k$};
		
		\draw[-{Latex[length=2mm]},thick] (-2, -0.5) -- (2, -0.5-1.1) node[midway, above, align=center] {$x(t_k)$}; 
		\draw[thick] (2, -0.5-1.1) -- (2, -0.5-1.1-0.7);  
		\draw[-{Latex[length=2mm]},thick] (2, -0.5-1.1-0.7) -- (-2, -0.5-1.1-0.7-0.6)  node[midway, above, align=center] {$u(\cdot)$};; 
		
		\draw [bracket] (2.3,-0.5) --  (2.3, -1.6)  node [midway, right=5pt, align=left] {Uplink\\[-0.25em]delay $\tau_\mathrm{UL}(t_k)$}; 
		\draw [bracket] (2.3, -1.6) --  (2.3, -2.3)  node [midway, right=5pt, align=left] {Controller\\[-0.25em]calculation $\tau_\mathrm{C}$};
		\draw [bracket] (2.3, -2.3) --  (2.3, -2.9)  node [midway, right=5pt, align=left] {\baselineskip=0pt Downlink\\[-0.25em]delay $\tau_\mathrm{DL}$};
		
		\draw[thick, color=red] (-2.1, -4.0) -- (-1.9, -4.0) node[above left, left=15pt, color=red, align=right] {Event-trigger\\[-0.25em]$\scriptsize\lVert x(t_{k+1}) - \hat{x}(t_{k+1}) \rVert \geq \varepsilon$};
		\draw[-{Latex[length=2mm]},thick, color=red] (-2, -4.0) -- (2, -4.0-1.1); 
		\draw[thick, color=red] (2, -4.0-1.1) -- (2, -4.0-1.1-0.7);  
		\draw[-{Latex[length=2mm]},thick, color=red] (2, -4.0-1.1-0.7) -- (-2, -4.0-1.1-0.7-0.6)  node[midway, above, align=center] {}; 
		\draw[thick, color=blue] (-2.1, -5.4) -- (-1.9, -5.4) node[above left, left=15pt, color=blue, align=right] {Input timeout $\scriptsize t\geq$\\[-0.25em]$\scriptsize t_k+ \tau_{RTT}(t_k) + T_\mathrm{com} - \bar{\tau}_\mathrm{RTT}$};
		\draw[-{Latex[length=2mm]},thick, color=blue] (-2, -5.4) -- (2, -5.4-1.1); 
		\draw[thick, color=blue] (2, -5.4-1.1) -- (2, -5.4-1.1-0.7);  
		\draw[-{Latex[length=2mm]},thick, color=blue] (2, -5.4-1.1-0.7) --(-2, -5.4-1.1-0.7-0.6)  node[midway, above, align=center] {};  

		\draw [bracket] (-1.7, -2.9) -- (-1.7, -9.4)  node [midway, right=5pt, fill=white, align=left] {Input\\[-0.25em]trajectory\\[-0.25em]length $T_\mathrm{com}$};

		\draw [bracket] (-2.3, -9.4) -- (-2.3, -5.4)  node [midway, left=5pt, align=right] {Round-trip\\[-0.25em]bound$\bar{\tau}_{RTT}$};
		
		\draw [thick]  (-2.1, -9.4) -- (-1.9, -9.4) node [left=15pt, align=right] {End of input trajectory\\$t_k + \tau_{RTT}(t_k) + T_\mathrm{com}$};

	\end{tikzpicture}
	}
	\caption{Message sequence diagram of the networked controller.
    Two cases can trigger an event, either as depicted in red the state prediction error exceeds a threshold $\varepsilon$, or if the available control inputs are about to expire in $\bar{\tau}_\mathrm{RTT}$ depicted in blue.}
    \label{fig:message_sequence_diagram}
\end{figure}

The networked control Algorithm \ref{alg:networked_control_system} is an adaptation of \cite[Algorithm 1]{varutti2009event} for a deterministic downlink and controller delay.
The controller solves the continuous-time Optimal Control Problem (OCP)
\begin{subequations}
	\begin{align}
		\min_{u(\cdot)} &\int_{0}^{T} \ell(x(t),u(t)) \;\mathrm{d}t + V_{\mathrm{f}}(x(T)) \\
		\text{s.t.}  \quad &\forall t \in [0, T] \nonumber \\
		& \dot{x}(t) = f(x(t), u(t)), \quad x(0) = x_0 \\
		& u(t) \in \mathbb{U} \\
		& x(t) \in \mathbb{X} \\
		& x(T) \in \mathbb{X}_{\mathrm{f}}. \label{eq:ocp_Xf}
	\end{align} 
	\label{eq:ocp}%
\end{subequations}
The OCP objective consists of the stage cost $\ell: \mathbb{X} \times \mathbb{U} \rightarrow \mathbb{R}^+$ and terminal penalty $V_\mathrm{f}: \mathbb{X} \rightarrow \mathbb{R}^+$, which we design in Section \ref{sec:system_model_and_ocp_design}.
The controller then sends the first part of the optimal input trajectory $u^\star(t)$ for all $t \in \left[0,T_{\text{com}}\right]$, where  $T_{\text{com}}$ is the communicated horizon.

\renewcommand{\algorithmicindent}{1em}
\begin{algorithm}
	\caption{Networked Control System Algorithm}
    \label{alg:networked_control_system}    
    \textbf{Preparation}\\
   	Measure downlink $\tau_{\mathrm{DL}}$ and controller calculation $\tau_{\mathrm{C}}$ delays
	\textbf{Sensor}
	\begin{algorithmic}[1]
		\FOR{each sampling time $t_k$}
		\STATE Measure state $x(t_k)$
		\STATE Determine the maximum $t_{\mathrm{u,max}}$ duration of any available input trajectory
		\STATE Compute $\hat{x}(t_k)$ based on the latest input trajectory $(u(\cdot), t_x)$ and the corresponding state measurement $x(t_x)$
		\IF{ $t_{k+1} \geq t_{\mathrm{u,max}} + \bar{\tau}_\mathrm{RTT}$ or $\lVert x(t_k) - \hat{x}(t_k) \rVert \geq \varepsilon$ }
		\STATE Send the packet $(x(t_k), t_k)$ to the controller
		\STATE Save the current state measurement $(x_k, t_k)$
		\ENDIF
		\ENDFOR
	\end{algorithmic}
	\textbf{Controller}
	\begin{algorithmic}[1]
		\STATE \texttt{control\_input\_buffer} $= \{(u^\star(\cdot), t_0 + \tau_{RTT}(t_0))\}$
		\STATE $t_{\mathrm{old}} = t_0$
		\WHILE{True}
		\STATE Wait for the packet $(x(t_k), t_k)$
		\IF{$t_k > t_{\mathrm{old}}$}
		\STATE $t_{\mathrm{old}} = t_k$
		\STATE Calculate the expected arrival time $\tau_{\mathrm{RTT}}(t_k)\!=\!t\!-\!t_k\!+\!\tau_{\mathrm{C}}\!+\!\tau_{\mathrm{DL}}$
		\STATE Apply the delay compensation \eqref{eq:delay_compensation} to obtain $\hat{x}(t_k + \tau_{\mathrm{RTT}}(t_k))$
		\STATE Solve OCP~\eqref{eq:ocp} with $x_0 = \hat{x}(t_k + \tau_{\mathrm{RTT}}(t_k))$
		\STATE Extract the first part of the optimal input trajectory $u(\cdot) \leftarrow u^\star(t), t \in \left[t_k +  \tau_{\mathrm{RTT}}(t_k), t_k +  \tau_{\mathrm{RTT}}(t_k) + T_{com} \right]$
		\STATE Send the first part of the input, the state timestamp, and the expected arrival time to the actuator as $(u(\cdot), t_k, t_k +  \tau_{\mathrm{RTT}}(t_k))$
		\STATE Append $(u(\cdot), t_k\!+\!\tau_{\mathrm{RTT}}(t_k))$ to the \texttt{control\_input\_buffer}
		\ENDIF
		\ENDWHILE
	\end{algorithmic}
	\textbf{Actuator}
	\begin{algorithmic}[1]
		\WHILE{True}
		\STATE Receive the input trajectory packet $(u(\cdot), t_k,  t_k\!+\!\tau_{\mathrm{RTT}}(t_k))$
		\STATE Start applying the newly received control input trajectory $u(\cdot)$ at $t_k\!+\!\tau_{\mathrm{RTT}}(t_k)$
		\STATE Provide the input trajectory to the sensor $(u(\cdot), t_k, t_k + t_{\mathrm{RTT}}(t_k))$
		\ENDWHILE
	\end{algorithmic}
\end{algorithm}

To initialize the delay compensation, the control input applied to the system before the first package is received needs to be known, which can be practically achieved, e.g., by starting with a zero input on both controller and actuator side.

\subsection{Nominal Stability Analysis}
We now turn towards the nominal stability guarantees of the proposed scheme. The next result can be found in {\cite{varutti2009event}}.
\begin{theorem}[Nominal Stability {\cite[Theorem 4.1]{varutti2009event}}]
    \label{thm:stability_with_terminal_constraint}
	Consider Algorithm \ref{alg:networked_control_system} using OCP \eqref{eq:ocp} with terminal penalty $V_{\mathrm{f}}(x)$ and terminal region $\mathbb{X}_{\mathrm{f}}$ such that
	\begin{enumerate}[i)]
		\item $V_{\mathrm{f}} \in C^1$, $V_{\mathrm{f}}(0) = 0$, and $\mathbb{X}_{\mathrm{f}} \subset X$ is closed, connected and contains the origin
		\item There exist $T$ such that $0 < \delta \leq t_{i+1} - t_{i} < T$, for all $i \in \mathbb{N}$ and some $\delta \in \mathbb{R}^+$
            \item For all $k$, $\tau_{RTT}(t_k )\leq \bar{\tau}_\mathrm{RTT}$ 
		\item For all $x_0 \in \mathbb{X}_{\mathrm{f}}$, there exists $u(\tau) \in U$, $\tau \in [0, T]$ such that
		\begin{subequations}
			\begin{align}
				&x(\tau) \in \mathbb{X}_{\mathrm{f}}, \\
				&\dot{x}(\tau) = f(x(\tau), u(\tau)), \quad x(0) = x_0, \\
				&\frac{\partial V_{\mathrm{f}}}{\partial x} f(x(\tau), u(\tau)) + \ell(x(\tau), u(\tau)) \leq 0 
			\end{align}
		\end{subequations}
		\item OCP~\eqref{eq:ocp} is feasible at $t_0$.
	\end{enumerate}
	Then $\lim_{t\rightarrow \infty} \lVert x(t)\rVert \rightarrow 0$.
\end{theorem}

We now extend the previous analysis to derive stability without terminal region.
\begin{theorem}[Nominal Stability without Terminal Region]
    Consider Algorithm \ref{alg:networked_control_system} using OCP \eqref{eq:ocp} such that stability is guaranteed by Theorem \ref{thm:stability_with_terminal_constraint} for the terminal region 
    $\mathbb{X}_{\mathrm{f}}  = \left\{x \in \mathbb{X} \,\middle|\, V_{\mathrm{f}}(x) \leq \gamma \right\}$ 
    with $\gamma > 0$. 
    Then, for any initial condition from which $x=0$ can be reached in finite time, there exists $\beta < \infty$, such that the OCP without terminal constraint $\mathbb{X}_{\mathrm{f}} = \mathbb{X}$ and weighted terminal penalty $V_{\mathrm{f},\beta}(x) = \beta V_{\mathrm{f}}(x)$ guarantees asymptotic convergence, i.e. $\lim_{t\rightarrow \infty} \lVert x(t)\rVert \rightarrow 0$. 
\end{theorem}
\begin{proof}
Due to space limitations we only give a brief sketch of the proof. 
The NMPC scheme is an adaptation of Algorithm~\ref{alg:networked_control_system}. A straight-forward continuous-time extension of \cite[Theorem 3]{limon2006stability} shows convergence without delays, whereby, if $\beta$ is chosen sufficiently large, the terminal constraint~\eqref{eq:ocp_Xf} is satisfied without being explicitly considered. Closed-loop convergence follows by Theorem~\ref{thm:stability_with_terminal_constraint}.
\end{proof}
In the nominal case, where the system model equals the plant and no disturbances are present, the predicted and actual states will be identical. 
Hence, no events will be triggered by the event-triggering mechanism, and thus, the nominal stability analysis is independent of the event-triggering condition and the event-triggering threshold. 
For event-triggered MPC with bounded disturbances, we refer to {\cite{7922495,li2014event, 10141876}}.

\section{EXPERIMENTS}
\label{sec:experiments}
Next, we discuss the experiments for the proposed event-triggered NMPC controller. 
We begin by presenting the experimental setup, which comprises a rotational inverted pendulum and our tailored O-RAN 6G research network. This is followed by a detailed analysis of the experimental results.

\subsection{System Model and OCP Design}
\label{sec:system_model_and_ocp_design}
The control plant used in our experiments is a rotational inverted pendulum, specifically a Furuta pendulum~\cite{furuta1992}, the Quanser Servo 2, which is one of the most popular control benchmarks \cite{doi:10.5772/55058}. 
This system consists of a motor-driven rotary arm to which a pendulum link is attached, see Figure~\ref{fig:furuta_sketch}. The state vector $x=(\theta, \alpha, \dot{\theta}, \dot{\alpha})^\top \in \mathbb{R}^4$ comprises the rotary arm angle $\theta  \in \left[-\pi/2, \pi/2\right]$, the pendulum arm angle $\alpha$, and their respective angular velocities $\dot{\theta}$ and $\dot{\alpha}$. 
The angular positions $\theta$ and $\alpha$ are obtained from incremental encoders, the velocity of the rotary pendulum arm $\dot{\theta}$ is obtained from a digital tachometer, and the pendulum arm velocity $\dot{\alpha}$ is obtained from fixed-position differences of the pendulum arm incremental encoder.
The link lengths are denoted by $L_p$ and $L_r$. The control input $u \in \left[-7.5, 7.5\right]$ is the voltage applied to the motor that drives the rotary arm.
\begin{figure}
    \centering
    \begin{subfigure}[b]{0.32\columnwidth}
        \centering
         \includegraphics[width=\textwidth]{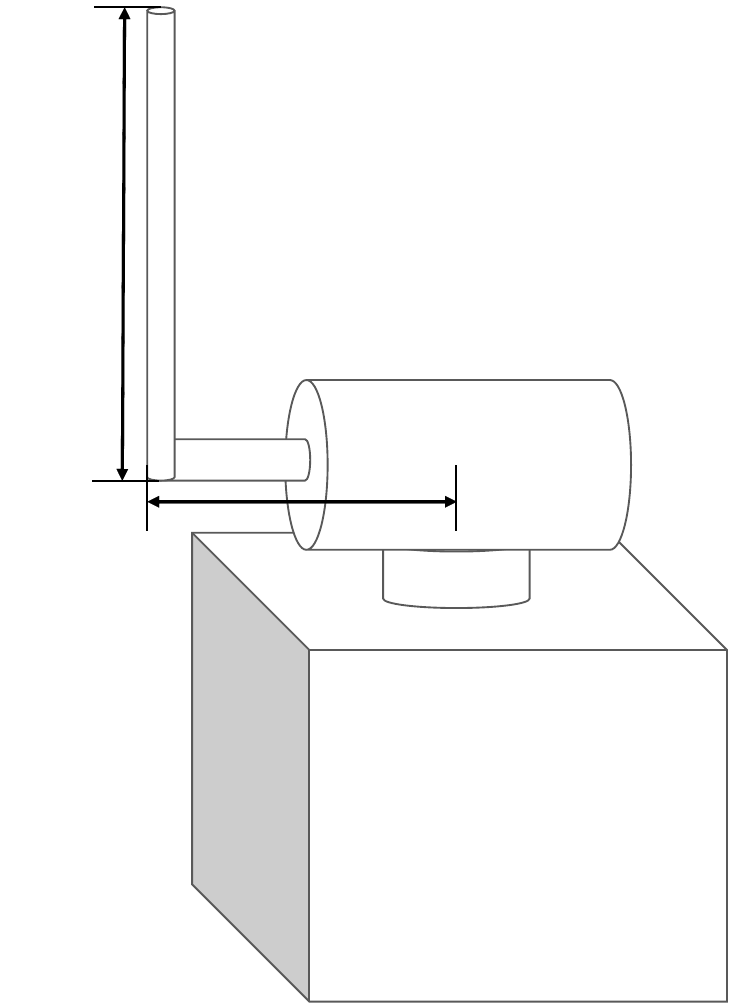}
        \caption{Side view}
        \label{fig:furuta_swing}
    \end{subfigure}
    \hfill
    \begin{subfigure}[b]{0.32\columnwidth}
        \centering
        \includegraphics[width=0.95\textwidth]{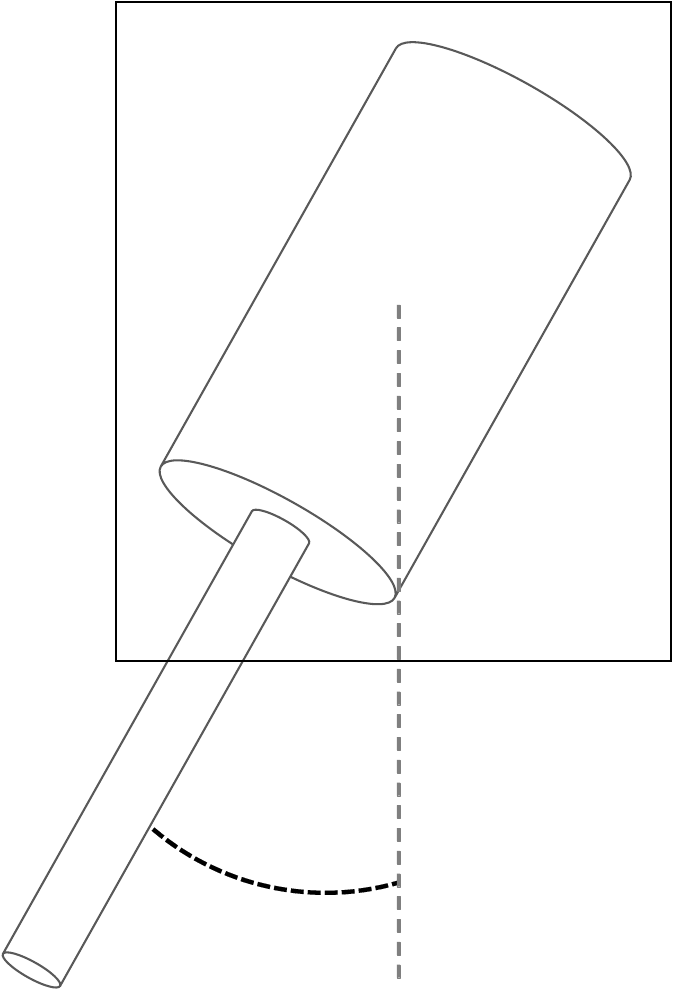}
        \caption{Top view}
        \label{fig:furuta_top}
    \end{subfigure}
    \hfill
    \begin{subfigure}[b]{0.32\columnwidth}
        \centering
        \def\svgwidth{0.7\textwidth}
        \includegraphics[width=0.7\textwidth]{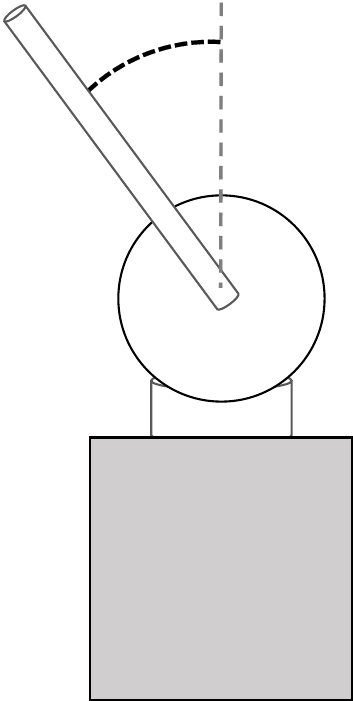}
        \caption{Front view}
        \label{fig:furuta_front}
    \end{subfigure}
    \caption{Furuta pendulum schematic with length $L_r$ and angle $\theta$ for the rotary arm as well as length $L_p$ and angle $\alpha$ for the pendulum arm.}
    \label{fig:furuta_sketch}
\end{figure}

The Furuta pendulum dynamics $\dot{x} = f(x, u)$ are derived similar to \cite{furuta1992} and the resulting nonlinear equations of motion read
\renewcommand{\sin}{\operatorname{s}}
\renewcommand{\cos}{\operatorname{c}}
\begin{subequations}
    \begin{equation}
    \begin{split}
        \left( m_p L_r^2 + \frac{1}{4}m_pL_p^2 \sin^2(\alpha) + J_r \right) \ddot{\theta} 
        - \frac{1}{2}m_pL_pL_r\cos(\alpha)\ddot{\alpha} &\\+ \frac{1}{2}m_pL_pL_r\sin(\alpha)\dot{\alpha}^2 
        + \frac{1}{2}m_pL_p^2\sin(\alpha)\cos(\alpha)\dot{\alpha}\dot{\theta} = \tau - B_r\dot{\theta},
    \end{split}
    \end{equation}
    \begin{equation}
    \begin{split}
        -\frac{1}{2}m_pL_pL_r\cos(\alpha) \ddot{\theta} + \left( J_p + \frac{1}{4}m_pL_p^2 \right) \ddot{\alpha} &\\
        - \frac{1}{4}m_pL_p^2\cos(\alpha)\sin(\alpha)\dot{\theta}^2 
        - \frac{1}{2}m_pL_pg\sin(\alpha) = -B_p\dot{\alpha},
    \end{split}
    \end{equation}
	\begin{equation}
        \tau = \frac{k_t \left(u - k_m \cdot \dot{\theta}\right)}{R_m}.
	\end{equation}
    \label{eq:eom}%
\end{subequations}
We use $s(\cdot)$ and $c(\cdot)$ to denote the sine and cosine function respectively.
The parameters and their values are given in Table~\ref{tab:system_parameters}.\\

\begin{table}[h!]
	\centering
	\caption{Parameters of the Furuta Pendulum.}
	\label{tab:system_parameters}
	\begin{tabular}{ccl}
		\toprule
		\textbf{Parameter} & \textbf{Value} & \textbf{Description} \\ 
		\midrule
		$R_m$              & 8.4 $\Omega$          & Resistance \\
		$k_t$              & 0.042 N$\cdot$m/A     & Current-torque constant \\
		$k_m$              & 0.042 V$\cdot$s/rad   & Velocity constant \\
		$m_r$              & 0.095 kg              & Mass of rotary arm \\
		$L_r$              & 0.085 m               & Total length of rotary arm \\
		$J_r$              & $5.72 \times 10^{-5}$ kg$\cdot$m$^2$ & Moment of inertia of rotary arm \\
		$D_r$              & $2.7 \times 10^{-4}$ m$\cdot$s/rad & Viscous damping of rotary arm \\
		$m_p$              & 0.024 kg              & Mass of pendulum link \\
		$L_p$              & 0.129 m               & Total length of pendulum link \\
		$J_p$              & $3.33 \times 10^{-5}$ kg$\cdot$m$^2$ & Moment of inertia of pendulum link \\
		$D_p$              & $5 \times 10^{-5}$ N$\cdot$m$\cdot$s/rad & Viscous damping of pendulum link \\
		$g$                & 9.81 m/s$^2$          & Gravitational acceleration \\
		\bottomrule
	\end{tabular}
\end{table}

We design OCP \eqref{eq:ocp} to stabilize the system at its upper equilibrium $x_\text{goal}=(0, 0, 0, 0)^\top$, starting from its lower equilibrium $x_0=(0, -\pi, 0, 0)^\top$. This requires a swing-up from the initial position. The stage cost and terminal penalty are defined as $\ell(x(t), u(t)) = 1/2\, x(t)^\top Q x(t) + 1/2\, R u(t)^2$ and $V_\mathrm{f}(x) = 1/2\, x^\top P x$, respectively, where $Q = \text{diag}(1.6, 1.6, 0.1, 0.01)$, $R = 0.4$, and $P$ is the solution of the algebraic Riccati equation for the linearized dynamics at the upper equilibrium. We observe stability for $\beta=1$, since the prediction horizon is long enough to swing the pendulum up.
State and input constraints are imposed based on the physical limitations of the system, maximum motor torque, and rotary arm angle. The state constraints on the rotary arm angle are implemented as soft constraints. 
The slacks $s_{\mathrm{ub}}, s_{\mathrm{lb}}: [0, T] \rightarrow \mathbb{R}^{+}$ for the upper bound $\theta(t)-\pi/2 \leq s_{\mathrm{ub}}(t)$ and lower bound $ s_{\mathrm{lb}}(t) \geq \theta(t)+\pi/2$ are penalized within the stage cost by $\ell_{\mathrm{s}}(x) = 0.1 \cdot (s_{\mathrm{lb}}^2 + s_{\mathrm{ub}}^2) +  s_{\mathrm{lb}} +  s_{\mathrm{ub}}$. 
For real-time computation, the resulting OCP is implemented in the Acados framework~\cite{Verschueren2021}. 
We use a prediction horizon of $T= \SI{2}{\second}$ with $N = 50$ stages and a variable sampling time. The first 15 stages are discretized with a finer sampling interval of $\Delta t=\SI{20}{\milli\second}$, while the remaining trajectory uses a coarser discretization with $\Delta t=\SI{42,5}{\milli\second}$. 
The length of the communicated trajectories $T_{\text{com}} = \SI{300}{\milli\second}$ is designed to match the finer discretization of the first 15 stages. 
To trigger control law recalculations before the currently available input trajectories expire, a round-trip time bound of $\bar{\tau}_{\mathrm{RTT}}=\SI{100}{\milli\second}$ is used.

Finally, we use the closed-loop cost $J_{\text{CL}}$ as a metric to evaluate control performance over the experiment duration $T_{\text{CL}}$ in our experiments:
\begin{align}
    J_{\mathrm{CL}}=\frac{1}{T_{\mathrm{CL}}}\int_{0}^{T_{\mathrm{CL}}} \ell(x(t),u(t)) \; \mathrm{d}t.
    \label{eq:clc}
\end{align}

\subsection{Communication Setup}
\label{sec:comm_setup}
For the real-world O-RAN communication testbed, we employ an open 5G stack as a basis, allowing for flexible adaptation of cell parameters to the specific needs of control systems, including low latency and reduced jitter, i.e., variance of latency  \cite{overbeck2024data}.
Cellular networks, like 5G, are preferred over contention-based networks, such as Wi-Fi, due to their reliably low latencies even under high network loads \cite{9446078, Arendt2024DistributedPerformance}.

The 5G system is realized using the O-RAN stack \textit{srsRAN Project} \cite{srsRAN}, the core network Open5GS, and an NI USRP X310 Software-Defined Radio (SDR).
We employ a commercial-grade Quectel RM520Q-GL 5G modem to provide connectivity for the pendulum.
To lower average latency and jitter, we tune the reactive uplink scheduling strategy so that the modem can request uplink resources at the radio resource scheduling of the base station every \SI{4}{\milli\second}.
The adaption comes at the cost of higher total radio resources due to overhead by the high scheduling frequency.
To reduce latency induced by waiting times, we set tight timing requirements on the modem.
This includes the time between uplink resource grant and uplink data transmission, as well as the delay for downlink data acknowledgments.
On transport layer, we use the User Datagram Protocol (UDP).

To challenge the control algorithm with realistic channel conditions, we employ a Propsim F64 real-time radio channel emulator realizing a Tapped-Delay-Line C fading channel with a delay spread of 300 and maximum doppler frequency of \SI{100}{\hertz} as defined by the 3rd Gen. Partnership Project (3GPP) standardization in \cite{tr38.901}.
Additionally, we challenge the algorithm with different Signal-to-Noise-Ratios (SNRs) to force packet errors of up to \SI{80}{\percent}, resulting in either higher latency in case of packet re-transmission by the Hybrid-Automatic-Repeat Request (HARQ) mechanism of 5G or result in packet loss if undetected.
Hence, we can provide realistic scenarios including far cell-edge deployments of the control application.
For reproducible conditions, we setup an SNR with variable interference power distributed as Additive White Gaussian Noise (AWGN) to force a given SNR, resulting in bit errors.
Hence, the modulation and coding is successively reduced by the channel adaptation mechanisms to reduce packet failures up to the most robust possible coding, after which errors inevitably appear.
Before, the rapidly changing channel defined by the channel model leads to minor errors.
\begin{figure*}
        \centering
        \includegraphics[width=\textwidth, trim=0cm 0.2cm 0cm 0.2cm, clip]{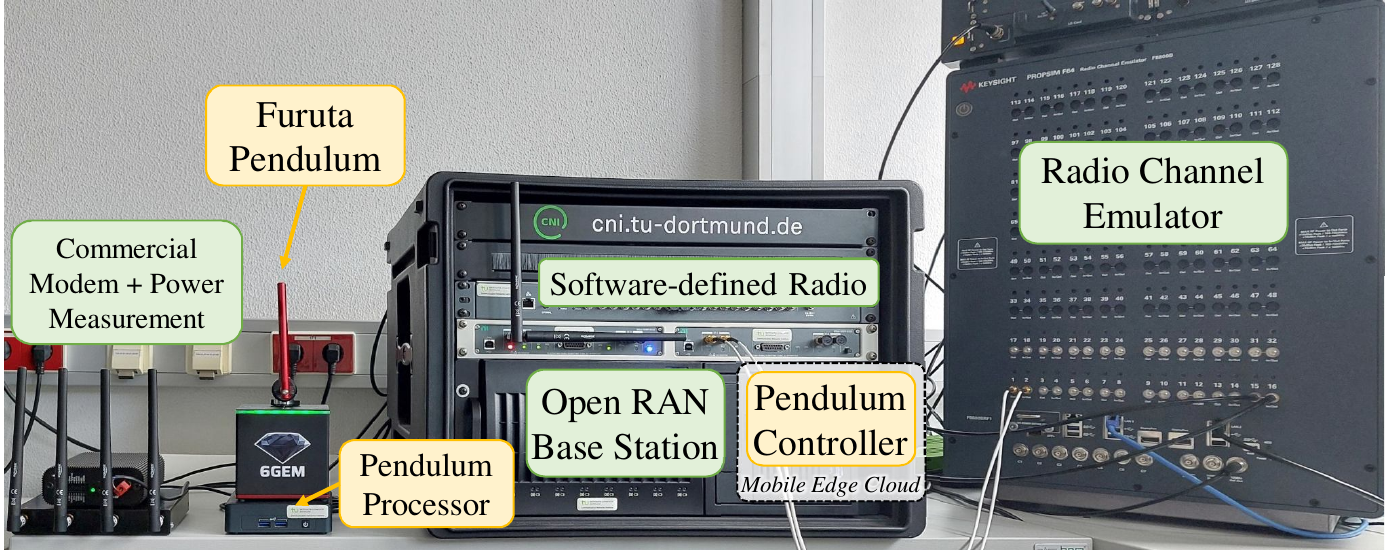}
        \caption{Experimental setup comprising the Furuta pendulum and the O-RAN 6G communication research platform.}
        \label{fig:exp_setup}
\end{figure*}

\subsection{Experimental Results} 
We analyze the trade-off between closed-loop cost~{\eqref{eq:clc}} and the communication demand relative to periodic control for varying event-triggering thresholds $\epsilon$. 
We also investigate the impact of the varying channel conditions on the delay pattern and the resulting closed-loop control performance.
Before the experiments are started, we measure the sum of the downlink and controller calculation delays by averaging delays using a channel without additional noise resulting in $\tau_{\mathrm{DL}}+\tau_{\mathrm{C}}=\SI{9.31}{\milli\second}$. This value is used as reference within the delay compensation for the constant delays.
The results are shown in Figure~\ref{fig:results_eps} and~\ref{fig:results_snr}, respectively. 
Each bar shows the average value for a total of five measurements, where each run consists of swinging up and stabilizing the Furuta pendulum for an experiment duration of $T_{\mathrm{CL}}= \SI{30}{\second}$.\\
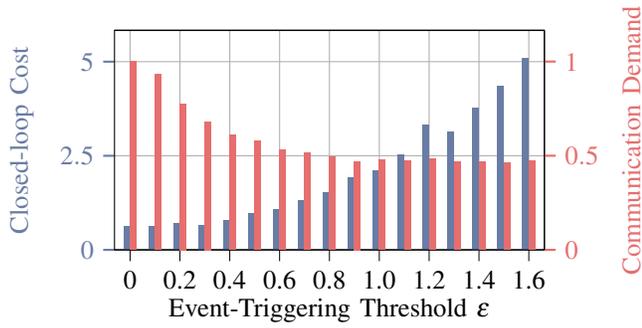
\begin{figure}
\begin{tikzpicture}

\definecolor{darkgray176}{RGB}{176,176,176}
\definecolor{lightslategray105124165}{RGB}{105,124,165}
\definecolor{salmon230110110}{RGB}{230,110,110}

\begin{axis}[
    width=2.25in,
    height=1.75in,
    scale only axis,
    tick align=outside,
    tick pos=left,
    x grid style={darkgray176},
    xlabel={Event-Triggering Threshold $\epsilon$},
    xmin=-0.5, xmax=16.75,
    xtick style={color=black},
    xtick={0.125,2.125,4.125,6.125,8.125,10.125,12.125,14.125,16.125},
    xticklabels={0,0.2,0.4,0.6,0.8,1.0,1.2,1.4,1.6},
    y grid style={darkgray176},
    ylabel=\textcolor{lightslategray105124165}{Closed-loop Cost},
    ymin=0, ymax=4.18022707176845,
    ytick={0,1,2,3,4},
    ytick style={color=lightslategray105124165, thick},
    yticklabel style={color=lightslategray105124165},
    axis line style={-},
    grid=both
]
\draw[draw=none,fill=lightslategray105124165] (axis cs:-0.125,0) rectangle (axis cs:0.125,0.58605796247618);
\draw[draw=none,fill=lightslategray105124165] (axis cs:0.875,0) rectangle (axis cs:1.125,0.615105056421133);
\draw[draw=none,fill=lightslategray105124165] (axis cs:1.875,0) rectangle (axis cs:2.125,0.619881728413371);
\draw[draw=none,fill=lightslategray105124165] (axis cs:2.875,0) rectangle (axis cs:3.125,0.636956440233698);
\draw[draw=none,fill=lightslategray105124165] (axis cs:3.875,0) rectangle (axis cs:4.125,0.711430367227703);
\draw[draw=none,fill=lightslategray105124165] (axis cs:4.875,0) rectangle (axis cs:5.125,0.835314906538979);
\draw[draw=none,fill=lightslategray105124165] (axis cs:5.875,0) rectangle (axis cs:6.125,1.02737299698927);
\draw[draw=none,fill=lightslategray105124165] (axis cs:6.875,0) rectangle (axis cs:7.125,1.10064937027196);
\draw[draw=none,fill=lightslategray105124165] (axis cs:7.875,0) rectangle (axis cs:8.125,1.32491485323863);
\draw[draw=none,fill=lightslategray105124165] (axis cs:8.875,0) rectangle (axis cs:9.125,1.48249395426867);
\draw[draw=none,fill=lightslategray105124165] (axis cs:9.875,0) rectangle (axis cs:10.125,1.87585644481786);
\draw[draw=none,fill=lightslategray105124165] (axis cs:10.875,0) rectangle (axis cs:11.125,2.01428277813516);
\draw[draw=none,fill=lightslategray105124165] (axis cs:11.875,0) rectangle (axis cs:12.125,2.26510627871564);
\draw[draw=none,fill=lightslategray105124165] (axis cs:12.875,0) rectangle (axis cs:13.125,2.75575775589674);
\draw[draw=none,fill=lightslategray105124165] (axis cs:13.875,0) rectangle (axis cs:14.125,2.99758430212803);
\draw[draw=none,fill=lightslategray105124165] (axis cs:14.875,0) rectangle (axis cs:15.125,3.98116863977948);
\draw[draw=none,fill=lightslategray105124165] (axis cs:15.875,0) rectangle (axis cs:16.125,3.71612694116321);

\end{axis}

\begin{axis}[
    width=2.25in,
    height=1.75in,
    scale only axis,
    axis y line=right,
    tick align=outside,
    x grid style={darkgray176},
    xmin=-0.5, xmax=16.75,
    xtick=\empty,   
    xticklabels={},
    y grid style={darkgray176},
    ylabel=\textcolor{salmon230110110}{\shortstack{Communication Demand}},
    ymin=0, ymax=1.045,
    ytick={0,0.25,0.5,0.75,1},
    ytick pos=right,
    ytick style={color=salmon230110110, thick},
    yticklabel style={anchor=west, color=salmon230110110},
    axis line style={-}
]
\draw[draw=none,fill=salmon230110110] (axis cs:0.125,0) rectangle (axis cs:0.375,1);
\draw[draw=none,fill=salmon230110110] (axis cs:1.125,0) rectangle (axis cs:1.375,0.887301753683741);
\draw[draw=none,fill=salmon230110110] (axis cs:2.125,0) rectangle (axis cs:2.375,0.754057514265786);
\draw[draw=none,fill=salmon230110110] (axis cs:3.125,0) rectangle (axis cs:3.375,0.632702896555837);
\draw[draw=none,fill=salmon230110110] (axis cs:4.125,0) rectangle (axis cs:4.375,0.554996341544117);
\draw[draw=none,fill=salmon230110110] (axis cs:5.125,0) rectangle (axis cs:5.375,0.513128488032729);
\draw[draw=none,fill=salmon230110110] (axis cs:6.125,0) rectangle (axis cs:6.375,0.477633359437728);
\draw[draw=none,fill=salmon230110110] (axis cs:7.125,0) rectangle (axis cs:7.375,0.446020667544924);
\draw[draw=none,fill=salmon230110110] (axis cs:8.125,0) rectangle (axis cs:8.375,0.454330411990144);
\draw[draw=none,fill=salmon230110110] (axis cs:9.125,0) rectangle (axis cs:9.375,0.418523776035649);
\draw[draw=none,fill=salmon230110110] (axis cs:10.125,0) rectangle (axis cs:10.375,0.430496203097702);
\draw[draw=none,fill=salmon230110110] (axis cs:11.125,0) rectangle (axis cs:11.375,0.42438247158107);
\draw[draw=none,fill=salmon230110110] (axis cs:12.125,0) rectangle (axis cs:12.375,0.428533136526705);
\draw[draw=none,fill=salmon230110110] (axis cs:13.125,0) rectangle (axis cs:13.375,0.419932463093032);
\draw[draw=none,fill=salmon230110110] (axis cs:14.125,0) rectangle (axis cs:14.375,0.437055602919395);
\draw[draw=none,fill=salmon230110110] (axis cs:15.125,0) rectangle (axis cs:15.375,0.454346291585991);
\draw[draw=none,fill=salmon230110110] (axis cs:16.125,0) rectangle (axis cs:16.375,0.428110888222981);
\end{axis}
\end{tikzpicture}
    \caption{Closed-loop cost and communication demand for varying event-triggering thresholds $\epsilon$.}
    \label{fig:results_eps}
\end{figure}

First, the pendulum is controlled without additional noise by the channel emulator to analyze performance under varying event-triggering thresholds $\epsilon$.
A threshold of $\epsilon=0$ corresponds to periodic control and, therefore, a communication demand of \SI{100}{\percent}. 
With increasing $\epsilon$, the recalculation of the closed-loop solution becomes less frequent, and the open-loop control input is applied for a longer duration, resulting in a higher closed-loop cost, as illustrated in Figure~\ref{fig:results_eps}. 
On the other hand, this reduces the frequency of control signal transmissions over the wireless interface, thereby lowering the communication demand. Depending on the application requirements, we observe the best trade-off between control performance and communication demand. 
For the Furuta pendulum, we consider $\epsilon=0.4$ as a suitable threshold. In this case, the communication demand is reduced to \SI{55}{\percent} network traffic with a negligible impact on the control performance.\\
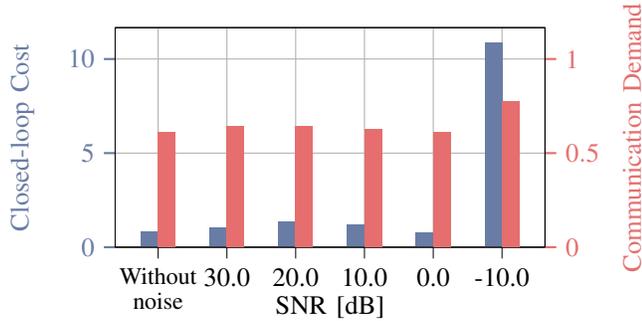
\begin{figure}
    \centering
\begin{tikzpicture}

\definecolor{darkgray176}{RGB}{176,176,176}
\definecolor{lightslategray105124165}{RGB}{105,124,165}
\definecolor{salmon230110110}{RGB}{230,110,110}

\begin{axis}[
    width=2.25in,
    height=1.75in,
    scale only axis,
    tick align=outside,
    tick pos=left,
    x grid style={darkgray176},
    xlabel={\shortstack{\\\\\\\\SNR [dB]}},
    xmin=-0.5, xmax=10,
    xtick style={color=black},
    xtick={0.25,1.25,3.25,5.25,7.25,9.25},
    xticklabels={\rotatebox{-45}{\shortstack{\small Without\\ \small noise}},\rotatebox{-45}{30.0}, \rotatebox{-45}{20.0}, \rotatebox{-45}{10.0}, \rotatebox{-45}{0.0}, \rotatebox{-45}{-10.0}},
    y grid style={darkgray176},
    ylabel=\textcolor{lightslategray105124165}{Closed-loop Cost},
    ymin=0, ymax=8,
    ytick={0,2,4,6,8},
    ytick style={color=lightslategray105124165, thick},
    yticklabel style={color=lightslategray105124165},
    axis line style={-},
    grid=both
]
\draw[draw=none,fill=lightslategray105124165] (axis cs:0,0) rectangle (axis cs:0.25,0.711430367227703); 
\draw[draw=none,fill=lightslategray105124165] (axis cs:1,0) rectangle (axis cs:1.25,1.37984938794285);
\draw[draw=none,fill=lightslategray105124165] (axis cs:2,0) rectangle (axis cs:2.25,1.04746851368718);
\draw[draw=none,fill=lightslategray105124165] (axis cs:3,0) rectangle (axis cs:3.25,1.38820510311843);
\draw[draw=none,fill=lightslategray105124165] (axis cs:4,0) rectangle (axis cs:4.25,1.17927822905778);
\draw[draw=none,fill=lightslategray105124165] (axis cs:5,0) rectangle (axis cs:5.25,1.43833141974055);
\draw[draw=none,fill=lightslategray105124165] (axis cs:6,0) rectangle (axis cs:6.25,1.13750138533708);
\draw[draw=none,fill=lightslategray105124165] (axis cs:7,0) rectangle (axis cs:7.25,1.0832939536699);
\draw[draw=none,fill=lightslategray105124165] (axis cs:8,0) rectangle (axis cs:8.25,1.45208422675124);
\draw[draw=none,fill=lightslategray105124165] (axis cs:9,0) rectangle (axis cs:9.25,7.19196867968429); 
\end{axis}

\begin{axis}[
    width=2.25in,
    height=1.75in,
    scale only axis,
    axis y line=right,
    tick align=outside,
    x grid style={darkgray176},
    xmin=-0.5, xmax=10,
    xtick=\empty,   
    xticklabels={},
    y grid style={darkgray176},
    ylabel=\textcolor{salmon230110110}{\shortstack{Communication Demand}},
    ymin=0, ymax=1,
    ytick={0,0.25,0.5,0.75,1},
    ytick pos=right,
    ytick style={color=salmon230110110, thick},
    yticklabel style={anchor=west, color=salmon230110110},
    axis line style={-}
]
\draw[draw=none,fill=salmon230110110] (axis cs:0.25,0) rectangle (axis cs:0.5,0.554996341544117); 
\draw[draw=none,fill=salmon230110110] (axis cs:1.25,0) rectangle (axis cs:1.5,0.616731476345795);
\draw[draw=none,fill=salmon230110110] (axis cs:2.25,0) rectangle (axis cs:2.5,0.62901057928716);
\draw[draw=none,fill=salmon230110110] (axis cs:3.25,0) rectangle (axis cs:3.5,0.634830873676968);
\draw[draw=none,fill=salmon230110110] (axis cs:4.25,0) rectangle (axis cs:4.5,0.618741432623309);
\draw[draw=none,fill=salmon230110110] (axis cs:5.25,0) rectangle (axis cs:5.5,0.6007514766841);
\draw[draw=none,fill=salmon230110110] (axis cs:6.25,0) rectangle (axis cs:6.5,0.587931757218206);
\draw[draw=none,fill=salmon230110110] (axis cs:7.25,0) rectangle (axis cs:7.5,0.601018621425357);
\draw[draw=none,fill=salmon230110110] (axis cs:8.25,0) rectangle (axis cs:8.5,0.615755600936555);
\draw[draw=none,fill=salmon230110110] (axis cs:9.25,0) rectangle (axis cs:9.5,0.70983704031127);
\end{axis}

\end{tikzpicture}
    \caption{Closed-loop cost and communication demand over varying channel signal-to-noise ratio.}
    \label{fig:results_snr}
\end{figure}

Next, we keep $\epsilon=0.4$ constant while employing the radio channel emulator to provide a challenging communication channel and force packet errors via a decreasing SNR towards far cell-edge. 
The results are shown in Figure~\ref{fig:results_snr}. 
The closed-loop cost and communication demand are comparable for SNR values above $\SI{0}{\decibel}$. 
However, at \SI{-10}{\decibel}, increased channel interference leads to significant packet errors.
If the receiver of a message cannot decode the packet, it sends a negative-acknowledgment to the transmitter of the message, which performs an HARQ retransmission. 
From a control perspective, if such an HARQ retransmission occurs in the uplink direction, i.e., from the plant to the controller, the uplink delay increases and will be measured and accounted for by the delay compensation.
Meanwhile, if the packet error occurs in the downlink direction from the controller to the actuator, the increased downlink delay is not being correctly compensated by the delay compensation, since the downlink delay is assumed constant.

This is reflected in the increased closed-loop cost. To understand the factors influencing the controller performance, we analyze the delays at an SNR of \SI{-10}{\decibel}. Figure~\ref{fig:result_trajectories} presents the results for SNR=\SI{-10}{\decibel} compared to the case without added noise, illustrating the RTT, input trajectory $u(\cdot)$, and the state $\alpha(\cdot)$, which represents the angle of the pendulum arm. The input constraints $\bar{u}$ and $\underline{u}$ are shown as dashed lines.

Examining $\alpha$ with a setpoint of 0, the results indicate that, after the swing-up phase, the system remains stable when no additional noise is present on the channel. For SNR=\SI{-10}{\decibel}, the increased delays prolong the swing-up phase. Despite this, the system maintains stability even with latency peaks exceeding \SI{100}{\milli\second}. However, when multiple correlated latency peaks occur at around $t$=\SI{8}{\second}, the controller fails to maintain stability, causing the pendulum arm to fall and leading to an increased closed-loop cost.\\
\begin{figure}
    \begin{subfigure}{\columnwidth}
        \centering
        \input{figures/input_traj_snr-10.tex}
        \label{fig:alph_traj}
    \end{subfigure}
    \begin{subfigure}{\columnwidth}
        \centering
        \input{figures/alpha_traj_snr-10}
        \label{fig:rtt_traj}
    \end{subfigure}
    \begin{subfigure}{\columnwidth}
        \centering
        \input{figures/rtt_traj_snr-10}
        \label{fig:input_traj}
    \end{subfigure}
    \begin{tikzpicture}[overlay, remember picture]
		\node[] (textnode) at (5.75,7.25) {\parbox{2.5cm}{\centering Instability due to\\higher latency}};
		\draw[-{Latex[length=2mm]},thick] (textnode.west) -- (3.9,7.25);
		\draw[-{Latex[length=2mm]}, thick] (textnode.south) -- (3.75,4);
        \node[] (textnode2) at (3.35,5.75) {\parbox{1cm}{\centering Swing\\up}};
		\draw[-{Latex[length=2mm]},thick] (textnode2.west) -- (2.25,6);
	\end{tikzpicture}
    \caption{Closed loop trajectories of the pendulum arm $\alpha(\cdot)$, the input $u(\cdot)$, and the RTT for SNR$=\SI{-10}{\decibel}$ compared to the case without additional noise.}
    \label{fig:result_trajectories}
\end{figure}

\begin{figure}[!t]
	\centering
\begin{tikzpicture}

\definecolor{crimson2143940}{RGB}{214,39,40}
\definecolor{darkgray176}{RGB}{176,176,176}
\definecolor{darkorange25512714}{RGB}{255,127,14}
\definecolor{darkturquoise23190207}{RGB}{23,190,207}
\definecolor{forestgreen4416044}{RGB}{44,160,44}
\definecolor{goldenrod18818934}{RGB}{188,189,34}
\definecolor{gray127}{RGB}{127,127,127}
\definecolor{lightgray204}{RGB}{204,204,204}
\definecolor{mediumpurple148103189}{RGB}{148,103,189}
\definecolor{orchid227119194}{RGB}{227,119,194}
\definecolor{sienna1408675}{RGB}{140,86,75}
\definecolor{steelblue31119180}{RGB}{31,119,180}

\begin{groupplot}[
    width=1.05*\columnwidth/2,
	height=1.05*\columnwidth/2,
	group style={group size=2 by 2,horizontal sep=1cm, vertical sep=1.5cm}]
\nextgroupplot[
legend cell align={left},
legend columns=2,
legend style={
  fill opacity=0.8,
  draw opacity=1,
  text opacity=1,
  at={(1.2,1.15)},
  anchor=south,
  draw=lightgray204
},
tick align=outside,
tick pos=left,
x grid style={darkgray176},
xlabel={Uplink Delay [ms]},
xmajorgrids=false,
xmin=3, xmax=150,
xtick style={color=black},
y grid style={darkgray176},
ylabel={CDF},
ylabel near ticks,
ymajorgrids,
ymin=-0.05, ymax=1.05,
ytick style={color=black},
xmode=log,
xtick={4,5,6,7,8,9,10,20,30,40,50,60,70,80,90,100},
xticklabels={4,,,,,,10,,,,,,,,100},
extra x ticks={4,10,100},
extra x tick style={grid=major},
extra x tick labels={}
]
\addplot [semithick, steelblue31119180, const plot mark left]
table {figures/networked_snr_cdf_delays/fig-000_reduced.dat};
\addlegendentry{SNR=-10.0}
\addplot [semithick, darkorange25512714, const plot mark left]
table {figures/networked_snr_cdf_delays/fig-001_reduced.dat};
\addlegendentry{SNR=-5.0}
\addplot [semithick, forestgreen4416044, const plot mark left]
table {figures/networked_snr_cdf_delays/fig-002_reduced.dat};
\addlegendentry{SNR=0.0}
\addplot [semithick, crimson2143940, const plot mark left]
table {figures/networked_snr_cdf_delays/fig-003_reduced.dat};
\addlegendentry{SNR=5.0}
\addplot [semithick, mediumpurple148103189, const plot mark left]
table {figures/networked_snr_cdf_delays/fig-004_reduced.dat};
\addlegendentry{SNR=10.0}
\addplot [semithick, sienna1408675, const plot mark left]
table {figures/networked_snr_cdf_delays/fig-005_reduced.dat};
\addlegendentry{SNR=15.0}
\addplot [semithick, orchid227119194, const plot mark left]
table {figures/networked_snr_cdf_delays/fig-006_reduced.dat};
\addlegendentry{SNR=20.0}
\addplot [semithick, gray127, const plot mark left]
table {figures/networked_snr_cdf_delays/fig-007_reduced.dat};
\addlegendentry{SNR=25.0}
\addplot [semithick, goldenrod18818934, const plot mark left]
table {figures/networked_snr_cdf_delays/fig-008_reduced.dat};
\addlegendentry{SNR=30.0}
\addplot [semithick, darkturquoise23190207, const plot mark left]
table {figures/networked_snr_cdf_delays/fig-009_reduced.dat};
\addlegendentry{Without noise}

\nextgroupplot[
scaled y ticks=manual:{}{\pgfmathparse{#1}},
tick align=outside,
tick pos=left,
x grid style={darkgray176},
xlabel={Calculation Time [ms]},
xmajorgrids=false,
xmin=1.5, xmax=10.5,
xtick style={color=black},
xtick={2,3,4,5,6,7,8,9,10},
xticklabels={2,,,,,,,10},
extra x ticks={2,10},
extra x tick style={grid=major},
extra x tick labels={},
y grid style={darkgray176},
ylabel={},
ymajorgrids,
ymin=-0.05, ymax=1.05,
ytick style={color=black},
yticklabels={}, 
xmode=log
]
\addplot [semithick, steelblue31119180, const plot mark left]
table {figures/networked_snr_cdf_delays/fig-010_reduced.dat};
\addplot [semithick, darkorange25512714, const plot mark left]
table {figures/networked_snr_cdf_delays/fig-011_reduced.dat};
\addplot [semithick, forestgreen4416044, const plot mark left]
table {figures/networked_snr_cdf_delays/fig-012_reduced.dat};
\addplot [semithick, crimson2143940, const plot mark left]
table {figures/networked_snr_cdf_delays/fig-013_reduced.dat};
\addplot [semithick, mediumpurple148103189, const plot mark left]
table {figures/networked_snr_cdf_delays/fig-014_reduced.dat};
\addplot [semithick, sienna1408675, const plot mark left]
table {figures/networked_snr_cdf_delays/fig-015_reduced.dat};
\addplot [semithick, orchid227119194, const plot mark left]
table {figures/networked_snr_cdf_delays/fig-016_reduced.dat};
\addplot [semithick, gray127, const plot mark left]
table {figures/networked_snr_cdf_delays/fig-017_reduced.dat};
\addplot [semithick, goldenrod18818934, const plot mark left]
table {figures/networked_snr_cdf_delays/fig-018_reduced.dat};
\addplot [semithick, darkturquoise23190207, const plot mark left]
table {figures/networked_snr_cdf_delays/fig-019_reduced.dat};

\nextgroupplot[
tick align=outside,
tick pos=left,
x grid style={darkgray176},
xlabel={Downlink Delay [ms]},
xmin=4, xmax=25,
xtick style={color=black},
xtick={5,6,7,8,9,10,20},
xticklabels={5,,,,,10,20},
extra x ticks={5,10,20},
extra x tick style={grid=major},
extra x tick labels={},
y grid style={darkgray176},
ylabel={CDF},
ylabel near ticks,
ymajorgrids,
ymin=-0.05, ymax=1.05,
ytick style={color=black},
xmode=log
]
\addplot [semithick, steelblue31119180, const plot mark left]
table {figures/networked_snr_cdf_delays/fig-020_reduced.dat};
\addplot [semithick, darkorange25512714, const plot mark left]
table {figures/networked_snr_cdf_delays/fig-021_reduced.dat};
\addplot [semithick, forestgreen4416044, const plot mark left]
table {figures/networked_snr_cdf_delays/fig-022_reduced.dat};
\addplot [semithick, crimson2143940, const plot mark left]
table {figures/networked_snr_cdf_delays/fig-023_reduced.dat};
\addplot [semithick, mediumpurple148103189, const plot mark left]
table {figures/networked_snr_cdf_delays/fig-024_reduced.dat};
\addplot [semithick, sienna1408675, const plot mark left]
table {figures/networked_snr_cdf_delays/fig-025_reduced.dat};
\addplot [semithick, orchid227119194, const plot mark left]
table {figures/networked_snr_cdf_delays/fig-026_reduced.dat};
\addplot [semithick, gray127, const plot mark left]
table {figures/networked_snr_cdf_delays/fig-027_reduced.dat};
\addplot [semithick, goldenrod18818934, const plot mark left]
table {figures/networked_snr_cdf_delays/fig-028_reduced.dat};
\addplot [semithick, darkturquoise23190207, const plot mark left]
table {figures/networked_snr_cdf_delays/fig-029_reduced.dat};

\nextgroupplot[
scaled y ticks=manual:{}{\pgfmathparse{#1}},
tick align=outside,
tick pos=left,
x grid style={darkgray176},
xlabel={RTT Delay [ms]}, 
xmin=9, xmax=150,
xtick style={color=black},
extra x ticks={10,100},
extra x tick style={grid=major},
extra x tick labels={},
y grid style={darkgray176},
ylabel={},
ymajorgrids,
ymin=-0.05, ymax=1.05,
ytick style={color=black},
yticklabels={}, 
xmode=log,
log ticks with fixed point,
]
\addplot [semithick, steelblue31119180, const plot mark left]
table {figures/networked_snr_cdf_delays/fig-030_reduced.dat};
\addplot [semithick, darkorange25512714, const plot mark left]
table {figures/networked_snr_cdf_delays/fig-031_reduced.dat};
\addplot [semithick, forestgreen4416044, const plot mark left]
table {figures/networked_snr_cdf_delays/fig-032_reduced.dat};
\addplot [semithick, crimson2143940, const plot mark left]
table {figures/networked_snr_cdf_delays/fig-033_reduced.dat};
\addplot [semithick, mediumpurple148103189, const plot mark left]
table {figures/networked_snr_cdf_delays/fig-034_reduced.dat};
\addplot [semithick, sienna1408675, const plot mark left]
table {figures/networked_snr_cdf_delays/fig-035_reduced.dat};
\addplot [semithick, orchid227119194, const plot mark left]
table {figures/networked_snr_cdf_delays/fig-036_reduced.dat};
\addplot [semithick, gray127, const plot mark left]
table {figures/networked_snr_cdf_delays/fig-037_reduced.dat};
\addplot [semithick, goldenrod18818934, const plot mark left]
table {figures/networked_snr_cdf_delays/fig-038_reduced.dat};
\addplot [semithick, darkturquoise23190207, const plot mark left]
table {figures/networked_snr_cdf_delays/fig-039_reduced.dat};
\end{groupplot}

\end{tikzpicture}
	\caption{Delay components for different SNR values.}
	\label{fig:delays}
\end{figure}
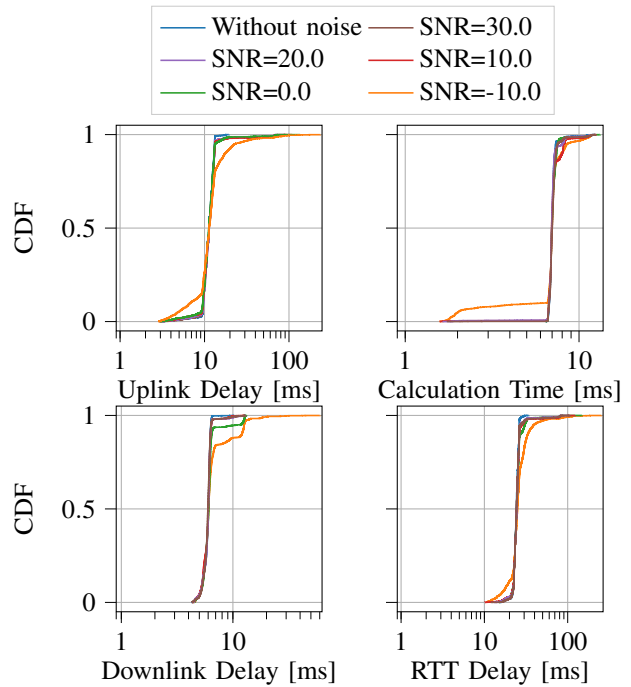

\begin{table}
	\centering
	\caption{Average of the delays for different SNR values.}
    \label{tab:average_delays_snr}
    \begin{tabular}{lcccc}
\toprule
 SNR [dB] & \makecell[c]{Uplink\\Delay [ms]} & \makecell[c]{Calculation\\Time [ms]} & \makecell[c]{Downlink\\Delay [ms]} & \makecell[c]{RTT\\Delay [ms]} \\
\midrule
\makecell[l]{Without noise} & 11.06 & 4.61 & 4.80 & 20.47 \\
30.0 & 14.65 & 4.52 & 5.02 & 24.18 \\
25.0 & 14.86 & 4.83 & 5.13 & 24.81 \\
20.0 & 14.21 & 4.59 & 4.97 & 23.75 \\
15.0 & 13.57 & 4.59 & 4.99 & 23.15 \\
10.0 & 12.82 & 4.54 & 5.04 & 22.39 \\
5.0 & 11.70 & 4.57 & 5.06 & 21.33 \\
0.0 & 11.45 & 4.62 & 5.08 & 21.15 \\
-5.0 & 11.86 & 4.69 & 5.24 & 21.79 \\
-10.0 & 17.22 & 4.40 & 5.57 & 27.09 \\
\bottomrule
\end{tabular}

\end{table}

To analyze the delay distribution in more detail, we show the cumulative distribution function (CDF) of the delays for varying SNR values in Figure~\ref{fig:delays} and the corresponding average delays in Table~\ref{tab:average_delays_snr}.
For the settings without added noise and good SNR values of 30 to \SI{10}{\decibel}, the delays are centered closely around their average values since no to little significant transmission errors occur.
In the uplink direction under low signal level, resulting in an SNR of \SI{-10}{\decibel}, a large part of the available resources are allocated to the critical pendulum UE. Here, uplink resource over-allocations happen and the UE can send data without sending scheduling requests, resulting in lower-latency outliers.
These outliers below the average delay are at around \SI{6}{\milli\second}.
However, if multiple re-transmissions are needed to transmit a packet without errors, uplink delays of up to \SI{225}{\milli\second} are observed.
The problems with instability for the case of \SI{-10}{\decibel} also lead to an early termination of the OCP calculation and thus reduced calculation times.
In the downlink direction, no scheduling requests have to be made. Hence, the downlink delay can only vary within one slot length.

\section{CONCLUSIONS \& OUTLOOK}
\label{sec:conclusion}
This paper illustrates the potential of co-design of networked control via future 6G networks by providing first experimental results of event-triggered NMPC within a 6G research platform.
Even for the control benchmark of swinging up and stabilizing an inverted pendulum at its upper equilibrium, event-triggered control reduces communication demand to \SI{61}{\percent} compared to periodic control while maintaining similar control performance. 
While the networked control system is robust to delays up to \SI{100}{\milli\second}, prolonged delays due to poor channel conditions can affect the stability of the control application.
Future work will investigate the impact of event-triggered control on energy consumption and battery life in networked control systems. This can be explored using 5G RedCap (Reduced Capability) devices, which enable efficient communication with lower data rate requirements, thereby extending battery life \cite{jorke2024redcap}. These devices are able to traverse into a power-saving state for a long time during a period of inactivity. Specifically, a device keeps its context on cell side and can quickly wake up for data transmission to transmit data without establishing a new connection. An event-triggered control strategy could leverage this feature by keeping the device in an inactive state until a new event occurs. Ultimately, the co-design of the controller and the network should prioritize optimizing resource utilization, thereby minimizing energy consumption. Finally, further experiments with diverse control applications and communication systems are needed to advance networked control applications.

\printbibliography

\end{document}